\documentclass[pdflatex,sn-mathphys-num]{sn-jnl}

\usepackage{graphicx}%
\usepackage{multirow}%
\usepackage{amsmath,amssymb,amsfonts}%
\usepackage{amsthm}%
\usepackage{mathrsfs}%
\usepackage[title]{appendix}%
\usepackage{xcolor}%
\usepackage{textcomp}%
\usepackage{manyfoot}%
\usepackage{booktabs}%
\usepackage{algorithm}%
\usepackage{algorithmicx}%
\usepackage{algpseudocode}%
\usepackage{listings}%
\usepackage{mathtools}%
\usepackage{bm}%
\usepackage{hyperref}%
\usepackage{enumitem}%
\usepackage[nameinlink,noabbrev]{cleveref}

\theoremstyle{thmstyleone}%
\newtheorem{theorem}{Theorem}%
\newtheorem{proposition}{Proposition}%
\newtheorem{lemma}{Lemma}%
\newtheorem{corollary}{Corollary}%

\theoremstyle{thmstyletwo}%

\theoremstyle{thmstylethree}%
\newtheorem{definition}{Definition}%

\newcommand{\R}{\mathbb{R}}
\newcommand{\C}{\mathbb{C}}
\newcommand{\Lap}{\mathcal{L}}
\newcommand{\dd}{\,\mathrm{d}}
\newcommand{\ee}{\mathrm{e}}

\newcommand{\TricomiU}{U}

\newcommand{\CTD}{\,{}^{UC}\!D}
\newcommand{\RLTD}{\,{}^{URL}\!D}
\newcommand{\ITU}{\,{}^{U}\!I}

\begin{document}

\title[A canonical two-scale Sonine fractional calculus induced by the Tricomi function]%
{A canonical two-scale Sonine fractional calculus induced by the Tricomi function 
\footnote{Published in: {\bf Fract. Calc. Appl. Anal.} (2026). % 13540, 0556 (2026).
DOI: \href{https://doi.org/10.1007/s13540-026-00556-z}{10.1007/s13540-026-00556-z}.}}

\author*[1]{\fnm{Ivano} \sur{Colombaro}}\email{ivano.colombaro@unibz.it}

\author[2,3]{\fnm{Marc} \sur{Tudela-Pi}}\email{marc.tudela@csic.es}

\affil*[1]{\orgdiv{Faculty of Engineering}, \orgname{Free University of Bozen-Bolzano}, \orgaddress{\city{via Bruno Buozzi 1, Bolzano}, \country{Italy}}}

\affil[2]{\orgdiv{Institut de Microelectr\`onica de Barcelona (IMB-CNM)}, \orgname{CSIC}, \orgaddress{\city{Bellaterra}, \postcode{08193}, \country{Spain}}}

\affil[3]{\orgdiv{CIBER-BBN}, \orgname{Instituto de Salud Carlos III}, \orgaddress{\city{Madrid}, \country{Spain}}}

\abstract{
We introduce a Tricomi-type generalized fractional calculus in the Sonine kernel framework.
The key result is that the Tricomi branch is a Stieltjes function in the admissible parameter range, so its reciprocal is a complete Bernstein function.
This fact induces a Sonine fractional calculus together with the canonical Tricomi integral and the associated Riemann--Liouville-type and Caputo-type derivatives. We also prove that, within the Kummer class, the Tricomi branch is the unique Stieltjes representative, once the natural asymptotic normalization is fixed, and we derive the corresponding L\'evy--Khintchine representation, Volterra formulation, and scalar Cauchy problem. A distinctive feature of the resulting operators is the emergence of two independent asymptotic orders.
}



\maketitle

\section{Introduction} \label{sec:intro}

\setcounter{section}{1} \setcounter{equation}{0} %% to have proper 2-digits numbers of eqs
%% Note that this style produces 1-digit numbering of definitons, statements, exmaples, etc.

Generalized fractional calculus is naturally connected with completely monotone memory kernels, Stieltjes functions, complete Bernstein symbols and Sonine pairs. In this setting, nonlocal operators are encoded by the analytic structure of their Laplace symbols, the associated hereditary kernels and the inverse kernels defining the corresponding fractional integrals \cite{SchillingSongVondracek2012,Kochubei2011,LuchkoYamamoto2020,Luchko2021Sonine,Luchko2021Operational,Luchko2020Fundamental,giusti2020generalPrabhakar}.
A basic problem in this framework is therefore to identify explicit special-function laws that fit the Stieltjes/complete-Bernstein scheme and at the same time yield tractable operators in the time domain.
The  foundational framework of these pairs dates back to the seminal work of Sonine \cite{Sonine1884}, who generalized Abel's integral equation by introducing a pair of kernels whose convolution equals $1$. For a comprehensive historical overview of the development of fractional calculus and Sonine original ideas within the nineteenth-century mathematical community, the reader is referred to \cite{Rogosin2021}. 

In this context, while classical Sonine pairs often interpolate a single fractional order 
between different regimes, a major challenge has been the explicit formulation of kernels 
capable of genuinely capturing two independent asymptotic behaviors. To provide a clear 
mathematical explanation of our main finding, the core mechanism of this paper relies on 
exploiting the unique analytical properties of the Tricomi confluent hypergeometric function. 
By proving its Stieltjes nature, we bridge the gap between special functions theory and 
operational calculus, providing an explicit, non-trivial Sonine pair where the slow and 
fast regimes are governed by two decoupled exponents.

Indeed, we show that one of the two classical linearly independent solutions of the Kummer equation~\cite{Abramowitz1965handbook}, known as Tricomi confluent hypergeometric function of the second kind $\TricomiU(a,b,z)$~\cite{tricomi1947funzioni}, provides such a law. More precisely, we consider $\TricomiU(a,b,s\tau)$ in the parameter range
\begin{equation}
a\in(0,1),\qquad b\in(1,2),\qquad \tau>0,
\label{eq:intro:box}
\end{equation}
and we introduce the associated raw Tricomi symbol, defined as the reciprocal of the Tricomi function in the Laplace domain, namely
\begin{equation}
\varphi_{a,b,\tau}(s):=\frac{1}{\TricomiU(a,b,s\tau)},
\qquad s>0.
\label{eq:intro:rawsymbol}
\end{equation}
The classical analytic properties of \(\TricomiU\), including its integral representation, analytic continuation, branch structure and asymptotics, are well known in the literature \cite{DLMF,Buchholz1969,Temme1996}. In the present parameter range \eqref{eq:intro:box}, the two asymptotic sectors of the raw symbol are
\begin{align}
    \varphi_{a,b,\tau}(s)&\sim (s\tau)^a
\qquad\qquad\qquad\quad s\to\infty\,,
\\
\varphi_{a,b,\tau}(s)&\sim
\frac{\Gamma(a)}{\Gamma(b-1)}(s\tau)^{b-1}
\qquad s\to0^+\,,
\end{align}
thus the model carries two effective orders, namely \(a\) in the fast regime
and \(b-1\) in the slow regime.

The main contribution of the paper is to show that the Tricomi branch generates a canonical generalized fractional calculus in the Sonine kernel framework and, specifically, we prove that the function \(s\mapsto \TricomiU(a,b,s\tau)\) is
Stieltjes in the admissible parameter range. It follows that its reciprocal
\(\varphi_{a,b,\tau}\) is a complete Bernstein function, and hence that
\(\varphi_{a,b,\tau}(s)/s\) is again Stieltjes. This provides the analytic mechanism from which the whole operatorial structure is derived. In this sense, the Tricomi symbol in \eqref{eq:intro:rawsymbol} acts as a structural bridge between complete Bernstein generators in the Laplace domain and genuine Sonine fractional calculus in the time domain, as proved later in this paper.

Once this core fact is established, the induced Sonine calculus follows canonically. The same Tricomi law yields the explicit kernel
\begin{equation}
\kappa_{a,b,\tau}(t)
=
\frac{1}{\Gamma(a)\tau}
\left(\frac{t}{\tau}\right)^{a-1}
\left(1+\frac{t}{\tau}\right)^{b-a-1},    
\end{equation}
whose Laplace transform is exactly \(\TricomiU(a,b,s\tau)\), i.e.
\begin{equation}\label{eq:U_laplace_kernel}
    U(a,b,s\tau)
     = \int_{0}^{\infty} e^{-st}\,\kappa_{a,b,\tau}(t)\,\mathrm{d}t,
    \qquad s>0.
\end{equation}
Its Sonine partner is the derivative kernel associated with \(\varphi_{a,b,\tau}(s)/s\). In this
way we obtain a Tricomi integral, a Riemann--Liouville-type Tricomi derivative,
and a Tricomi--Caputo derivative, together with the corresponding generalized
fundamental theorems of calculus in the modern Sonine-kernel setting
\cite{SamkoCardoso2003,Kochubei2011,LuchkoYamamoto2020,Luchko2021Sonine,Luchko2021Operational,Luchko2024Sonin,Hanyga2020,Luchko2020Fundamental,diethelm2020why}.

A second contribution is structural. Since the raw Tricomi symbol is a complete Bernstein function, it admits a L\'evy--Khintchine representation~\cite{SchillingSongVondracek2012,MainardiRogosin2006}. We derive this representation explicitly and identify the associated L\'evy density. This reveals the hereditary geometry of the model at the level of the jump kernel and shows that the same two-scale structure is visible in the symbol, in the L\'evy density, and in the derivative kernel.

A further contribution is a rigidity result inside the Kummer class. The confluent hypergeometric equation admits two linearly independent branches, namely the Kummer $M$-branch and the Tricomi $U$-branch~\cite{Oldham2009atlasfunctions}, but only the latter is compatible with the Stieltjes structure required by the Sonine framework. We prove that, within the Kummer class, the Tricomi branch is the unique Stieltjes representative, up to the natural asymptotic normalization. In this sense, the normalized Tricomi law is not merely an explicit example, but the canonical Stieltjes member of that class.
We also show that the parameters of the Tricomi family are asymptotically identifiable from the raw symbol and, equivalently, from the associated L\'evy density or hereditary kernel.

Finally, we formulate the natural Volterra evolution problem induced by the Tricomi Sonine pair, derive its resolvent representation, prove existence and uniqueness of mild solutions, and show that, under additional regularity, this Volterra formulation is equivalent to the scalar Cauchy problem associated with the Tricomi--Caputo derivative.
The bounded Tricomi law then appears naturally as the canonical bounded response associated with shifted generators~\cite{Kochubei2011,LuchkoYamamoto2020,Luchko2021Operational} and the classical Volterra framework~\cite{GripenbergLondenStaffans1990}. From a broader perspective, the resulting Volterra formulation fits naturally within the general theory of Caputo-type fractional differential equations,
where existence, uniqueness and stability of solutions are typically studied under minimal regularity assumptions on the data \cite{Diethelm2019}.

Throughout the paper, all functions are assumed to be causal whenever
convolution operators are considered. The Volterra formulation is developed in
the natural low-regularity setting \(L^1_{\mathrm{loc}}([0,\infty))\), whereas
the Caputo-type realization is considered for functions in
\(AC_{\mathrm{loc}}([0,\infty))\). Whenever Laplace-transform identities are
used, we additionally assume exponential order so that the relevant transforms
exist in a right half-plane. These conventions will not be repeated each time.

The paper is organized as follows. In Sect.~\ref{sec:stieltjes} we prove that the Tricomi function is a Stieltjes law and deduce that the reciprocal raw symbol is complete Bernstein, while in Sect.~\ref{sec:levy} we derive the L\'evy--Khintchine representation and identify the corresponding L\'evy density.
The following Sect.~\ref{sec:sonine} is devoted to the construction of the induced Sonine fractional calculus, including the associated Tricomi integral, the Riemann--Liouville-type and Caputo-type derivatives, and the generalized fundamental theorems of generalized fractional calculus.
In Sect.~\ref{sec:characterization} we characterize the Tricomi family within the Kummer class of the Sonine framework and Sect.~\ref{sec:ident} addresses asymptotic identifiability within this Tricomi family. Finally, Sect.~\ref{sec:model} formulates the natural Volterra evolution problem, derives its resolvent representation and well-posedness, relates it to the associated scalar Cauchy problem and recovers the bounded Tricomi law from shifted symbols.
Discussion and concluding remarks, including perspectives and possible extensions, are collected in Sect.~\ref{sec:discussion}.% and Sect.~\ref{sec:conclusions}.

\section{The Tricomi function as a Stieltjes law}\label{sec:stieltjes}

Throughout the paper we work in the parameter range specified in \eqref{eq:intro:box}, namely $a\in(0,1)$, $b\in(1,2)$ and $\tau>0$.
We also use the following standard facts from the theory of Stieltjes and
complete Bernstein functions \cite{SchillingSongVondracek2012,Kochubei2011}.

\begin{definition}
A function \(f:(0,\infty)\to(0,\infty)\) is called a Stieltjes function if it
admits a representation
\begin{equation}
f(s)=\frac{\alpha}{s}+\beta+\int_0^\infty \frac{1}{s+r}\,\mu(\dd r),
\label{eq:Stdef}
\end{equation}
where \(\alpha,\beta\ge0\) and \(\mu\) is a positive Radon measure on
\((0,\infty)\) satisfying
\[
\int_0^\infty \frac{1}{1+r}\,\mu(\dd r)<\infty.
\]
\end{definition}

\begin{definition}
A function \(\phi:(0,\infty)\to[0,\infty)\) is called a complete Bernstein
function if it admits a representation
\begin{equation}
\phi(s)=c_0+c_1 s+\int_0^\infty \frac{s}{s+r}\,\sigma(\dd r),
\label{eq:CBFdef}
\end{equation}
where \(c_0,c_1\ge0\) and \(\sigma\) is a positive Radon measure on
\((0,\infty)\) satisfying
\[
\int_0^\infty \frac{1}{1+r}\,\sigma(\dd r)<\infty.
\]
\end{definition}

\begin{proposition}\label{prop:StCBF}
If \(f\not\equiv0\) is a Stieltjes function, then \(1/f\) is a complete Bernstein function. Conversely, if \(\phi\not\equiv0\) is a complete Bernstein function, then \(1/\phi\) is a Stieltjes function.
\end{proposition}

\begin{proposition}\label{prop:StCBFs}
If \(\phi\) is a complete Bernstein function, then \(\phi(s)/s\) is a Stieltjes function. Conversely, if \(g\) is a Stieltjes function, then \(s\,g(s)\) is a complete Bernstein function.
\end{proposition}

We next recall the standard integral representation of the Tricomi confluent hypergeometric function of the second kind \cite{tricomi1947funzioni}
\begin{equation}
\TricomiU(a,b,z)
=
\frac{1}{\Gamma(a)}
\int_0^\infty
\ee^{-zt}t^{a-1}(1+t)^{b-a-1}\,\dd t,
\quad \Re\{z\}>0,\quad a>0,\quad b>1.
\label{eq:Uintegral}
\end{equation}

\begin{proposition}\label{prop:Ubasic}
Let \(a\in(0,1)\), \(b\in(1,2)\), and \(\tau>0\). Then the function
\[
s\mapsto \TricomiU(a,b,s\tau),\qquad s>0,
\]
is positive and completely monotone on \((0,\infty)\). Moreover,
\begin{equation}
\TricomiU(a,b,s\tau)\sim
\frac{\Gamma(b-1)}{\Gamma(a)}(s\tau)^{1-b},
\qquad s\to0^+,
\label{eq:Usmall}
\end{equation}
and
\begin{equation}
\TricomiU(a,b,s\tau)\sim
(s\tau)^{-a},
\qquad |s|\to\infty,\qquad \Re\{s\}>0.
\label{eq:Ularge}
\end{equation}
\end{proposition}

\begin{proof}
Positivity follows directly from \eqref{eq:Uintegral}. Moreover, after the change
of variables \(t=\tau r\), one gets
\[
\TricomiU(a,b,s\tau)
=
\frac{1}{\Gamma(a)\tau}
\int_0^\infty
\ee^{-st}
\left(\frac{t}{\tau}\right)^{a-1}
\left(1+\frac{t}{\tau}\right)^{b-a-1}\dd t.
\]
Hence \(s\mapsto \TricomiU(a,b,s\tau)\) is the Laplace transform of the positive kernel
\[
\kappa_{a,b,\tau}(t)
=
\frac{1}{\Gamma(a)\tau}
\left(\frac{t}{\tau}\right)^{a-1}
\left(1+\frac{t}{\tau}\right)^{b-a-1},
\qquad t>0,
\]
and complete monotonicity follows from Bernstein's theorem \cite{SchillingSongVondracek2012}. The large-argument asymptotics \eqref{eq:Ularge} are standard for the Tricomi function, while
\eqref{eq:Usmall} follows from the connection formula near the origin, where
the singular branch \(z^{1-b}\) dominates when \(1<b<2\) \cite{DLMF,Buchholz1969,Temme1996}.
\end{proof}

\begin{definition}[Raw Tricomi symbol]
For \(a\in(0,1)\), \(b\in(1,2)\), and \(\tau>0\), define
\begin{equation}
\varphi_{a,b,\tau}(s):=\frac{1}{\TricomiU(a,b,s\tau)},
\qquad s>0.
\label{eq:rawsymbol}
\end{equation}
\end{definition}

\begin{lemma}[Boundary values of the Tricomi function on the cut]
\label{lem:cut}
For \(y>0\), the principal branch of the Tricomi function satisfies
\begin{equation}\label{eq:Ucut}
\TricomiU(a,b,-y+i0)=A(y)+iB(y),
\end{equation}
where, if \(b\neq a+1\),
\begin{equation}\label{eq:Acut}
A(y)=c_1 M(a,b,-y)+c_2 y^{1-b}\cos\theta\,M(a-b+1,2-b,-y),
\end{equation}
while
\begin{equation}
B(y)=c_2 y^{1-b}\sin\theta\,M(a-b+1,2-b,-y),
\label{eq:Bcut}
\end{equation}
with
\begin{equation}
c_1=\frac{\Gamma(1-b)}{\Gamma(a-b+1)},
\qquad
c_2=\frac{\Gamma(b-1)}{\Gamma(a)},
\qquad
\theta=\pi(1-b).
\label{eq:coeffcut}
\end{equation}
Moreover, by the Kummer transformation,
\begin{equation}
M(a-b+1,2-b,-y)=\ee^{-y}M(1-a,2-b,y),
\label{eq:Kummerapp}
\end{equation}
and therefore \(B(y)<0\) for every \(y>0\).

In the special case \(b=a+1\), one has
\[
\TricomiU(a,a+1,z)=z^{-a},
\]
hence
\[
A(y)=y^{-a}\cos(\pi a),
\qquad
B(y)=-\,y^{-a}\sin(\pi a),
\]
and again \(B(y)<0\) for every \(y>0\).
\end{lemma}

\begin{proof}
For \(b\neq a+1\), the boundary-value representation follows from the standard connection formula for the Tricomi function across the cut together with the principal-branch choice on \(\C\setminus(-\infty,0]\) \cite{DLMF,Buchholz1969,Temme1996}. Using \eqref{eq:Kummerapp}, one obtains
\[
M(a-b+1,2-b,-y)=\ee^{-y}M(1-a,2-b,y).
\]
Since \(1-a>0\) and \(2-b>0\), the Maclaurin coefficients of
\(M(1-a,2-b,y)\) are positive, so the right-hand side is positive for every
\(y>0\). Because \(\sin\theta<0\) for \(1<b<2\), it follows that \(B(y)<0\).

\noindent If \(b=a+1\), then \eqref{eq:Uintegral} reduces to
\[
\TricomiU(a,a+1,z)
=
\frac{1}{\Gamma(a)}\int_0^\infty \ee^{-zt}t^{a-1}\,\dd t
=
z^{-a},
\qquad \Re\{z\}>0.
\]
On the principal branch,
\[
(-y+i0)^{-a}=y^{-a}\ee^{-i\pi a},
\]
which gives the stated formulas for \(A(y)\) and \(B(y)\). Since
\(a\in(0,1)\), one has \(\sin(\pi a)>0\), and therefore \(B(y)<0\).
\end{proof}

\begin{theorem}[Direct Stieltjes representation of the Tricomi function]
\label{thm:USt}
Let \(a\in(0,1)\), \(b\in(1,2)\) and \(\tau>0\). Then
\[
s\mapsto \TricomiU(a,b,s\tau),\qquad s>0,
\]
is a Stieltjes function. More precisely,
\begin{equation}
\TricomiU(a,b,s\tau)
=
\int_0^\infty \frac{\rho^U_{a,b,\tau}(r)}{s+r}\,\dd r,
\qquad s>0,
\label{eq:UStieltjes}
\end{equation}
where
\begin{equation}
\rho^U_{a,b,\tau}(r)
=
\frac{\Gamma(b-1)\sin(\pi(b-1))}{\pi\,\Gamma(a)}
(\tau r)^{1-b}\ee^{-\tau r}M(1-a,2-b,\tau r),
\qquad r>0.
\label{eq:rhoU}
\end{equation}
In particular, \(\rho^U_{a,b,\tau}(r)\ge0\) for every \(r>0\).
\end{theorem}

\begin{proof}
Let us set \(f(z)=\TricomiU(a,b,z)\), which is analytic in the slit plane \(\C\setminus(-\infty,0]\), let us fix
\(z\in\C\setminus(-\infty,0]\) and let us consider a standard keyhole contour \(\Gamma_{R,\varepsilon}\) around the negative real axis. By Cauchy's formula,
\[
f(z)=\frac{1}{2\pi i}\int_{\Gamma_{R,\varepsilon}}
\frac{f(\zeta)}{\zeta-z}\,\dd \zeta \,,
\]
and considering \eqref{eq:Ularge}, one has \(f(\zeta)=O(|\zeta|^{-a})\) as
\(|\zeta|\to\infty\) in the principal sector. Hence the contribution of the
large circle is \(O(R^{-a})\to0\), since \(a>0\). Likewise, for
\eqref{eq:Usmall} one has \(f(\zeta)=O(|\zeta|^{1-b})\) as \(\zeta\to0\), so the contribution of the small circle is \(O(\varepsilon^{2-b})\to0\), since
\(b<2\).

Passing to the limit, only the two sides of the cut remain. Writing the
boundary values on the cut as
\[
\TricomiU(a,b,-x+i0)=A(x)+iB(x),
\qquad
\TricomiU(a,b,-x-i0)=A(x)-iB(x),
\qquad x>0,
\]
we obtain
\[
f(z)
=
\frac{1}{2\pi i}\int_0^\infty
\frac{\TricomiU(a,b,-x-i0)-\TricomiU(a,b,-x+i0)}{x+z}\,\dd x
=
-\frac{1}{\pi}\int_0^\infty \frac{B(x)}{x+z}\,\dd x.
\]

At this point, according to Lemma \ref{lem:cut}, we obtain
\[
B(x)=\frac{\Gamma(b-1)}{\Gamma(a)}
x^{1-b}\sin(\pi(1-b))\,M(a-b+1,2-b,-x),
\]
and using the Kummer transformation \eqref{eq:Kummerapp}, this becomes
\[
B(x)=\frac{\Gamma(b-1)}{\Gamma(a)}
x^{1-b}\sin(\pi(1-b))\,\ee^{-x}M(1-a,2-b,x).
\]
Since \(1-a\in(0,1)\) and \(2-b\in(0,1)\), the Maclaurin series of
\(M(1-a,2-b,x)\) has positive coefficients for \(x>0\), hence
\(M(1-a,2-b,x)>0\). Moreover \(\sin(\pi(1-b))<0\) for \(1<b<2\), so \(B(x)<0\)
for all \(x>0\). Therefore
\[
\rho^U_{a,b,1}(x):=-\frac{1}{\pi}B(x)\ge0,
\qquad x>0,
\]
and
\[
\TricomiU(a,b,z)=\int_0^\infty \frac{\rho^U_{a,b,1}(x)}{x+z}\,\dd x.
\]
Finally, setting \(z=s\tau\) and changing variables \(x=\tau r\) yields
\eqref{eq:UStieltjes} and \eqref{eq:rhoU}.

It remains to verify the Stieltjes integrability condition. Near the origin,
\eqref{eq:rhoU} gives
\[
\rho^U_{a,b,\tau}(r)\sim C_0\,r^{1-b},
\qquad r\to0^+,
\]
which is integrable because \(1-b\in(-1,0)\). At infinity, the standard
large-argument asymptotics of the Kummer function imply
\[
M(1-a,2-b,\tau r)\sim
\frac{\Gamma(2-b)}{\Gamma(1-a)}\,\ee^{\tau r}(\tau r)^{\,b-a-1},
\qquad r\to\infty,
\]
hence
\[
\rho^U_{a,b,\tau}(r)\sim C_\infty\,r^{-a},
\qquad r\to\infty,
\]
and therefore
\[
\frac{\rho^U_{a,b,\tau}(r)}{1+r}\sim C_\infty\,r^{-a-1},
\qquad r\to\infty,
\]
which is integrable since \(a\in(0,1)\). Thus \(\rho^U_{a,b,\tau}(r)\,\dd r\)
is indeed a Stieltjes measure.

In the special case \(b=a+1\), formula \eqref{eq:rhoU} reduces to
\[
\rho^U_{a,a+1,\tau}(r)
=
\frac{\sin(\pi a)}{\pi}\,(\tau r)^{-a}\ee^{-\tau r}M(1-a,1-a,\tau r)
=
\frac{\sin(\pi a)}{\pi}\,\tau^{-a}r^{-a},
\]
because $M(a,a,z)=\ee^{z}$. This is exactly the standard Stieltjes density for $(s\tau)^{-a}$.
\end{proof}

\begin{corollary}\label{cor:rawCBF}
The raw Tricomi symbol \(\varphi_{a,b,\tau}\) is a complete Bernstein
function.
\end{corollary}

\begin{proof}
By Theorem \ref{thm:USt}, the function \(s\mapsto \TricomiU(a,b,s\tau)\) is Stieltjes. Since it is positive and nonzero on \((0,\infty)\), Proposition \ref{prop:StCBF} implies that its reciprocal
\[
\varphi_{a,b,\tau}(s)=\frac{1}{\TricomiU(a,b,s\tau)}
\]
is complete Bernstein.
\end{proof}

\begin{corollary}[Two-scale asymptotics of the raw symbol]
\label{cor:rawasym}
The raw Tricomi symbol satisfies
\begin{equation}
\varphi_{a,b,\tau}(s)\sim
\frac{\Gamma(a)}{\Gamma(b-1)}(s\tau)^{b-1},
\qquad s\to0^+,
\label{eq:rawsmall}
\end{equation}
and
\begin{equation}
\varphi_{a,b,\tau}(s)\sim
(s\tau)^a,
\qquad |s|\to\infty,\qquad \Re\{s\}>0.
\label{eq:rawlarge}
\end{equation}
Hence the associated raw generator has effective order \(b-1\) in the slow
regime and effective order \(a\) in the fast regime.
\end{corollary}

\begin{proof}
These formulas follow immediately from \eqref{eq:Usmall} and
\eqref{eq:Ularge}.
\end{proof}

The structural relevance of Corollary \ref{cor:rawasym} lies in the sharp transition between the two effective fractional orders. In fact, a single operator naturally possesses a dual-scale structure induced by the structural properties of the Tricomi function. In the slow regime $s \to 0^+$ (long-time), the operator behaves like a fractional derivative of order $b-1$, whereas in the fast regime $s \to \infty$ (short-time), it transitions to order $a$. Crucially, since $a \in (0,1)$ and $b \in (1,2)$, these two asymptotic orders are fully independent and decoupled, providing a flexible framework for modelling complex physical systems with two evolutionary scales.

\section{L\'evy--Khintchine representation and L\'evy density}\label{sec:levy}

Since \(\varphi_{a,b,\tau}\) is a complete Bernstein function, it admits a
L\'evy--Khintchine representation \cite{SchillingSongVondracek2012, MainardiRogosin2006} and in this
section we also identify the corresponding L\'evy density explicitly.

\begin{theorem}[L\'evy--Khintchine representation]\label{thm:LK}
There exists a unique positive Radon measure \(\nu_{a,b,\tau}\) on
\((0,\infty)\), satisfying
\begin{equation}
\int_0^\infty \min\{1,r\}\,\nu_{a,b,\tau}(\dd r)<\infty,
\label{eq:nu-int}
\end{equation}
such that
\begin{equation}
\varphi_{a,b,\tau}(s)
=
\int_0^\infty \bigl(1-\ee^{-sr}\bigr)\,\nu_{a,b,\tau}(\dd r),
\qquad s>0.
\label{eq:LKmeasure}
\end{equation}
Moreover, the drift coefficient vanishes.
\end{theorem}

\begin{proof}
By Corollary \ref{cor:rawCBF}, the function \(\varphi_{a,b,\tau}\) is complete
Bernstein, hence
\[
\varphi_{a,b,\tau}(s)
=
c_0+c_1 s+\int_0^\infty \bigl(1-\ee^{-sr}\bigr)\,\nu_{a,b,\tau}(\dd r),
\qquad s>0,
\]
for some \(c_0,c_1\ge0\) and some positive Radon measure \(\nu_{a,b,\tau}\) satisfying \eqref{eq:nu-int} \cite{SchillingSongVondracek2012}. Since \(\TricomiU(a,b,s\tau)\to\infty\) as \(s\to0^+\), one has
\[
c_0=\varphi_{a,b,\tau}(0^+)=0.
\]
Moreover,
\[
c_1=\lim_{s\to\infty}\frac{\varphi_{a,b,\tau}(s)}{s},
\]
and employing \eqref{eq:rawlarge},
\[
\frac{\varphi_{a,b,\tau}(s)}{s}\sim \tau^a s^{a-1}\to0,
\qquad s\to\infty,
\]
since \(a\in(0,1)\). Therefore \(c_0=c_1=0\) and \eqref{eq:LKmeasure} is proved.
\end{proof}
It is worth noting that the measure $\nu$, also known as the L\'evy measure, appearing in the L\'evy–Khintchine representation of a complete Bernstein function, is in one-to-one correspondence with the Stieltjes measure associated with $\varphi(s)/s$.
This correspondence is classical in the theory of complete Bernstein symbols and underlies their interpretation as Laplace exponents of subordinators~\cite{Kochubei2011}.

Let us now introduce the auxiliary function
\begin{equation}
\Psi_{a,b,\tau}(s)
:=
\frac{\varphi_{a,b,\tau}(s)}{s}
=
\frac{1}{s\,\TricomiU(a,b,s\tau)}.
\label{eq:Psidef}
\end{equation}

\begin{proposition}\label{prop:PsiSt}
The function \(\Psi_{a,b,\tau}\) is a Stieltjes function. More precisely, it admits a representation of the form
\begin{equation}
\Psi_{a,b,\tau}(s)
=
\int_0^\infty \frac{1}{x+s}\,\mu_{a,b,\tau}(\dd x),
\qquad s>0.
\label{eq:PsiSt-measure}
\end{equation}
for a positive Radon measure \(\mu_{a,b,\tau}\) on \((0,\infty)\). Moreover,
\(\mu_{a,b,\tau}\) is absolutely continuous on \((0,\infty)\), and its density
\(\eta_{a,b,\tau}\) is given by the Stieltjes--Perron formula
\begin{equation}
\eta_{a,b,\tau}(x)
=
-\frac{1}{\pi}\,
\Im\!\left[
\frac{1}{(-x+i0)\,\TricomiU(a,b,-\tau x+i0)}
\right],
\qquad x>0.
\label{eq:etaSP}
\end{equation}
\end{proposition}

\begin{proof}
By Corollary \ref{cor:rawCBF} and Proposition \ref{prop:StCBFs}, the quotient ${\varphi_{a,b,\tau}(s)}/{s}$ in the auxiliary function \eqref{eq:Psidef} is a Stieltjes function, so that $\Psi_{a,b,\tau}$ admits a representation
\[
\Psi_{a,b,\tau}(s)
=
\frac{\alpha}{s}+\beta+\int_0^\infty \frac{1}{x+s}\,\mu_{a,b,\tau}(\dd x),
\qquad s>0,
\]
with \(\alpha,\beta\ge0\).

We now show that \(\alpha=\beta=0\). By \eqref{eq:rawlarge}, we have
\[
\Psi_{a,b,\tau}(s)\sim \tau^a s^{a-1},
\qquad s\to\infty.
\]
Since \(a\in(0,1)\), it follows that \(\Psi_{a,b,\tau}(s)\to0\) as \(s\to\infty\),
and therefore \(\beta=0\). On the other hand, by \eqref{eq:rawsmall},
\[
\Psi_{a,b,\tau}(s)\sim
\frac{\Gamma(a)}{\Gamma(b-1)}\,\tau^{b-1}s^{b-2},
\qquad s\to0^+.
\]
Hence
\[
s\,\Psi_{a,b,\tau}(s)\sim
\frac{\Gamma(a)}{\Gamma(b-1)}\,\tau^{b-1}s^{b-1}\to0,
\qquad s\to0^+,
\]
because \(b-1\in(0,1)\). Thus \(\alpha=0\), and \eqref{eq:PsiSt-measure} follows.

Since \(s\mapsto \Psi_{a,b,\tau}(s)\) extends analytically to
\(\C\setminus(-\infty,0]\), the measure is absolutely continuous on
\((0,\infty)\), and \eqref{eq:etaSP} follows from Stieltjes--Perron inversion;
see \cite{SchillingSongVondracek2012}.
\end{proof}

Write the boundary values of the Tricomi function on the cut as
\begin{equation}
\TricomiU(a,b,-y+i0)=A(y)+iB(y),
\qquad y>0,
\label{eq:ABdef}
\end{equation}
where \(A\) and \(B\) are given explicitly in Lemma \ref{lem:cut}. Then
\eqref{eq:etaSP} becomes
\begin{equation}
\eta_{a,b,\tau}(x)
=
-\frac{B(\tau x)}{\pi\,x\,\bigl(A^2(\tau x)+B^2(\tau x)\bigr)}.
\label{eq:etaAB}
\end{equation}

\begin{theorem}[Explicit L\'evy density]\label{thm:mexplicit}
The measure \(\nu_{a,b,\tau}\) is absolutely continuous with respect to
Lebesgue measure,
\[
\nu_{a,b,\tau}(\dd r)=m_{a,b,\tau}(r)\,\dd r,
\]
where
\begin{equation}
m_{a,b,\tau}(r)
=
\int_0^\infty x\,\ee^{-xr}\,\eta_{a,b,\tau}(x)\,\dd x,
\qquad r>0.
\label{eq:mfrometa}
\end{equation}
Equivalently,
\begin{equation}
m_{a,b,\tau}(r)
=
\frac{1}{\pi\tau}
\int_0^\infty
\ee^{-yr/\tau}\,
\frac{-B(y)}{A^2(y)+B^2(y)}\,\dd y,
\qquad r>0.
\label{eq:mexplicit}
\end{equation}
\end{theorem}

\begin{proof}
From \eqref{eq:PsiSt-measure} we get
\begin{equation}
\varphi_{a,b,\tau}(s)
=
s\,\Psi_{a,b,\tau}(s)
=
\int_0^\infty \frac{s}{x+s}\,\mu_{a,b,\tau}(\dd x).
\label{eq:PhiMinusOnePsi}
\end{equation}
Since \(\mu_{a,b,\tau}\) is absolutely continuous on \((0,\infty)\), we may
write
\[
\mu_{a,b,\tau}(\dd x)=\eta_{a,b,\tau}(x)\,\dd x.
\]
Using now the elementary identity~\cite{SchillingSongVondracek2012}
\[
\frac{s}{x+s}
=
\int_0^\infty \bigl(1-\ee^{-sr}\bigr)\,x\,\ee^{-xr}\,\dd r,
\qquad x>0,\ s>0,
\]
substitution into \eqref{eq:PhiMinusOnePsi} and Tonelli's theorem yield
\[
\varphi_{a,b,\tau}(s)
=
\int_0^\infty \bigl(1-\ee^{-sr}\bigr)
\left(
\int_0^\infty x\,\ee^{-xr}\,\eta_{a,b,\tau}(x)\,\dd x
\right)\dd r.
\]
Comparing this latter with \eqref{eq:LKmeasure} proves \eqref{eq:mfrometa}, while \eqref{eq:mexplicit} follows from \eqref{eq:etaAB} after the change of
variables \(y=\tau x\).
\end{proof}

\begin{corollary}[Complete monotonicity]\label{cor:mcm}
The function \(m_{a,b,\tau}\) is completely monotone on \((0,\infty)\). In particular, \(m_{a,b,\tau}(r)\ge0\) and is decreasing.
\end{corollary}

\begin{proof}
By \eqref{eq:mfrometa}, the function \(m_{a,b,\tau}\) is the Laplace transform of the positive measure \(x\,\eta_{a,b,\tau}(x)\,\dd x\). Hence it is completely monotone. In particular,
\((-1)m'_{a,b,\tau}(r)\ge 0\) for all \(r>0\), so that \(m_{a,b,\tau}\) is decreasing~\cite{miller2001completelymonotonic}.
\end{proof}

\begin{corollary}[Scaling law]\label{cor:mscaling}
The L\'evy density satisfies
\begin{equation}
m_{a,b,\tau}(r)=\frac{1}{\tau}\,m_{a,b,1}\!\left(\frac{r}{\tau}\right),
\qquad r>0.
\label{eq:mscaling}
\end{equation}
\end{corollary}

\begin{proof}
This follows immediately from \eqref{eq:mexplicit}.
\end{proof}

\begin{proposition}[Asymptotics of the L\'evy density]\label{prop:masym}
The L\'evy density satisfies
\begin{equation}
m_{a,b,\tau}(r)\sim
\frac{a\,\tau^a}{\Gamma(1-a)}\,r^{-1-a},
\qquad r\to0^+,
\label{eq:msmall}
\end{equation}
and
\begin{equation}
m_{a,b,\tau}(r)\sim
\frac{\Gamma(a)}{\Gamma(b-1)}
\frac{(b-1)\tau^{b-1}}{\Gamma(2-b)}
\,r^{-b},
\qquad r\to\infty.
\label{eq:mlarge}
\end{equation}
\end{proposition}

\begin{proof}
From \eqref{eq:rawlarge} we have
\[
\varphi_{a,b,\tau}(s)\sim \tau^a s^a,
\qquad s\to\infty.
\]
Since \(0<a<1\) and \(m_{a,b,\tau}\ge0\), the Abelian--Tauberian correspondence
for Bernstein functions yields
\[
m_{a,b,\tau}(r)\sim
\frac{a\,\tau^a}{\Gamma(1-a)}\,r^{-1-a},
\qquad r\to0^+.
\]

Likewise, from \eqref{eq:rawsmall},
\[
\varphi_{a,b,\tau}(s)\sim
\frac{\Gamma(a)}{\Gamma(b-1)}\,\tau^{b-1}s^{b-1},
\qquad s\to0^+.
\]
Since \(b-1\in(0,1)\), the same correspondence gives
\[
m_{a,b,\tau}(r)\sim
\frac{\Gamma(a)}{\Gamma(b-1)}
\frac{(b-1)\tau^{b-1}}{\Gamma(2-b)}
\,r^{-b},
\qquad r\to\infty.
\]
For a deeper background on regular variation and Tauberian transfer, see~\cite{BinghamGoldieTeugels1989}.
\end{proof}

In conclusion of this section, it is worth noting that the singular behaviour of \(m_{a,b,\tau}\) near the origin is governed by the fast exponent \(a\), whereas its large-\(r\) tail is governed by the slow exponent \(b-1\). Thus, the two-scale structure of the raw symbol is already visible at the level of the jump (L\'evy) kernel.

\section{The induced Sonine fractional calculus}\label{sec:sonine}

The previous sections identify the raw Tricomi symbol $\varphi_{a,b,\tau}(s)$ in \eqref{eq:rawsymbol} as a complete Bernstein function and the quotient $\Psi_{a,b,\tau}$ in \eqref{eq:Psidef} as a Stieltjes function. We now show that these analytic facts induce a genuine
Sonine-type fractional calculus. For the general Sonine-kernel framework and
its operational and fundamental-theorem aspects, see
\cite{SamkoCardoso2003,Kochubei2011,LuchkoYamamoto2020,Luchko2021Sonine,Luchko2021Operational,Luchko2024Sonin,Hanyga2020,Luchko2020Fundamental,diethelm2020why}.

\subsection{The integral kernel}

The integral representation \eqref{eq:Uintegral} immediately yields an explicit
time-domain kernel associated with the Tricomi law.

\begin{proposition}\label{prop:kappa}
Let us consider the kernel
\begin{equation}
\kappa_{a,b,\tau}(t)
:=
\frac{1}{\Gamma(a)\tau}
\left(\frac{t}{\tau}\right)^{a-1}
\left(1+\frac{t}{\tau}\right)^{b-a-1},
\qquad t>0.
\label{eq:kappadef}
\end{equation}
Then \(\kappa_{a,b,\tau}\) is positive and locally integrable on
\((0,\infty)\), and
\begin{equation}
\Lap\{\kappa_{a,b,\tau}\}(s)=\TricomiU(a,b,s\tau),
\qquad s>0.
\label{eq:kappaLap}
\end{equation}
Moreover,
\begin{equation}
\kappa_{a,b,\tau}(t)\sim \frac{\tau^{-a}}{\Gamma(a)}\,t^{a-1},
\qquad t\to0^+,
\label{eq:kappa-small}
\end{equation}
and
\begin{equation}
\kappa_{a,b,\tau}(t)\sim \frac{\tau^{1-b}}{\Gamma(a)}\,t^{b-2},
\qquad t\to\infty.
\label{eq:kappa-large}
\end{equation}
\end{proposition}

\begin{proof}
Formula \eqref{eq:kappaLap} follows by rewriting \eqref{eq:Uintegral} with the
change of variables \(t=\tau r\):
\[
\TricomiU(a,b,s\tau)
=
\frac{1}{\Gamma(a)\tau}
\int_0^\infty
\ee^{-st}
\left(\frac{t}{\tau}\right)^{a-1}
\left(1+\frac{t}{\tau}\right)^{b-a-1}\dd t.
\]
The asymptotics follow immediately from \eqref{eq:kappadef}. Since $a-1\in(-1,0)$, the kernel is locally integrable near the origin.
\end{proof}

\subsection{The derivative kernel}

We now introduce the derivative kernel associated with the reciprocal Tricomi symbol,  which plays the role of the derivative kernel in the Sonine framework.

\begin{proposition}\label{prop:k}
Let \(k_{a,b,\tau}\) be defined by
\begin{equation}
\Lap\{k_{a,b,\tau}\}(s)
=
\Psi_{a,b,\tau}(s)
=
\frac{1}{s\,\TricomiU(a,b,s\tau)},
\qquad s>0.
\label{eq:kLap}
\end{equation}
Then \(k_{a,b,\tau}\) is completely monotone and locally integrable on
\((0,\infty)\). Moreover,
\begin{equation}
k_{a,b,\tau}(t)\sim
\frac{\tau^a}{\Gamma(1-a)}\,t^{-a},
\qquad t\to0^+,
\label{eq:ksmall}
\end{equation}
and
\begin{equation}
k_{a,b,\tau}(t)\sim
\frac{\Gamma(a)}{\Gamma(b-1)\Gamma(2-b)}
\,\tau^{b-1}t^{1-b},
\qquad t\to\infty.
\label{eq:klarge}
\end{equation}
\end{proposition}

\begin{proof}
By Proposition \ref{prop:PsiSt}, the function \(\Psi_{a,b,\tau}\) is Stieltjes. Hence it
is the Laplace transform of a completely monotone kernel \(k_{a,b,\tau}\).
Local integrability near the origin follows from \eqref{eq:ksmall}, since
\(0<a<1\).

The asymptotic relations follow from \eqref{eq:kLap},
\eqref{eq:rawlarge}, and \eqref{eq:rawsmall}. Indeed,
\[
\Psi_{a,b,\tau}(s)\sim \tau^a s^{a-1},
\qquad s\to\infty,
\]
and
\[
\Psi_{a,b,\tau}(s)\sim
\frac{\Gamma(a)}{\Gamma(b-1)}\,\tau^{b-1}s^{b-2},
\qquad s\to0^+.
\]
Standard Abelian--Tauberian arguments for Laplace transforms of completely monotone kernels yield \eqref{eq:ksmall} and \eqref{eq:klarge} \cite{BinghamGoldieTeugels1989,SchillingSongVondracek2012}.
\end{proof}

\subsection{Sonine identity}

The pair \((\kappa_{a,b,\tau},k_{a,b,\tau})\) is the natural candidate for a
Sonine pair, since their Laplace transforms multiply to \(1/s\).

\begin{theorem}[Sonine pair]\label{thm:sonine}
For every \(a\in(0,1)\), \(b\in(1,2)\), and \(\tau>0\), the kernels
\(\kappa_{a,b,\tau}\) and \(k_{a,b,\tau}\) satisfy
\begin{equation}
(\kappa_{a,b,\tau}*k_{a,b,\tau})(t)=1,
\qquad t>0,
\label{eq:Sonine}
\end{equation}
where
\[
(\kappa*k)(t):=\int_0^t \kappa(t-\xi)k(\xi)\,\dd\xi.
\]
Hence \((\kappa_{a,b,\tau},k_{a,b,\tau})\) is a Sonine pair.
\end{theorem}

\begin{proof}
By Propositions \ref{prop:kappa} and \ref{prop:k}, both kernels are locally integrable on
\((0,\infty)\), and their Laplace transforms exist for every \(s>0\). By the
Laplace convolution theorem,
\[
\Lap\{\kappa_{a,b,\tau}*k_{a,b,\tau}\}(s)
=
\Lap\{\kappa_{a,b,\tau}\}(s)\Lap\{k_{a,b,\tau}\}(s)
=
\TricomiU(a,b,s\tau)\cdot \frac{1}{s\,\TricomiU(a,b,s\tau)}
=
\frac{1}{s}.
\]
Since \(\Lap\{1\}(s)=1/s\), uniqueness of the Laplace transform yields
\eqref{eq:Sonine}.
\end{proof}

As a consequence, the Tricomi construction yields a genuine Sonine pair in the time domain, not only a formal factorization in Laplace space.
Although the kernel \(k_{a,b,\tau}\) does not admit a closed-form expression in general, its existence, complete monotonicity and asymptotic behaviour are fully characterized via its Laplace symbol.

\subsection{Tricomi integral and Tricomi derivatives}\label{ss:Tricomi-integral-derivatives}

We now introduce the integral and derivative operators induced by the Sonine pair \((\kappa_{a,b,\tau},k_{a,b,\tau})\) previously introduced and discussed. Following the standard Sonine-kernel framework,
this leads naturally to a Tricomi fractional integral together with
Riemann--Liouville-type and Caputo-type realizations of the associated generator.

\begin{definition}[Tricomi integral]
Let \(f\in L^1_{\mathrm{loc}}([0,\infty))\) be causal. We define
\begin{equation}
(\ITU_{a,b,\tau}f)(t)
:=
\int_0^t \kappa_{a,b,\tau}(t-\xi)\,f(\xi)\,\dd\xi,
\qquad t>0.
\label{eq:ITU}
\end{equation}
\end{definition}

\begin{definition}[Riemann--Liouville-type Tricomi derivative]
Let \(u\in L^1_{\mathrm{loc}}([0,\infty))\) be causal and let us assume that
\(k_{a,b,\tau}*u\in AC_{\mathrm{loc}}([0,\infty))\). We define
\begin{equation}
(\RLTD_{a,b,\tau}u)(t)
:=
\frac{\dd}{\dd t}\int_0^t k_{a,b,\tau}(t-\xi)\,u(\xi)\,\dd\xi.
\label{eq:RLTD}
\end{equation}
\end{definition}

\begin{definition}[Tricomi--Caputo derivative]
Let \(u\in AC_{\mathrm{loc}}([0,\infty))\) be causal and let us assume that \(u\) is
of exponential order. We define
\begin{equation}
(\CTD_{a,b,\tau}u)(t)
:=
\int_0^t k_{a,b,\tau}(t-\xi)\,u'(\xi)\,\dd\xi.
\label{eq:CTD}
\end{equation}
Equivalently, using integration by parts in Volterra form,
\begin{equation}
(\CTD_{a,b,\tau}u)(t)
=
\frac{\dd}{\dd t}
\int_0^t
k_{a,b,\tau}(t-\xi)\,[u(\xi)-u(0)]\,\dd\xi.
\label{eq:CTD2}
\end{equation}
\end{definition}

For clarity, here the superscripts $UC$ and $URL$ denote, respectively, the Caputo-type and Riemann–Liouville-type realizations of the Tricomi generator.

\begin{proposition}[Laplace-domain formulas]\label{prop:LaplaceOps}
Let \(f\in L^1_{\mathrm{loc}}([0,\infty))\) be causal and of exponential order,
and let \(u\in AC_{\mathrm{loc}}([0,\infty))\) be causal and of exponential
order. Then
\begin{equation}
\Lap\{\ITU_{a,b,\tau}f\}(s)
=
\TricomiU(a,b,s\tau)\,\widehat f(s),
\label{eq:Lap-I}
\end{equation}
\begin{equation}
\Lap\{\RLTD_{a,b,\tau}u\}(s)
=
\varphi_{a,b,\tau}(s)\,\widehat u(s),
\label{eq:Lap-RL}
\end{equation}
and
\begin{equation}
\Lap\{\CTD_{a,b,\tau}u\}(s)
=
\varphi_{a,b,\tau}(s)\widehat u(s)
-
\frac{\varphi_{a,b,\tau}(s)}{s}u(0).
\label{eq:Lap-Caputo}
\end{equation}
\end{proposition}

\begin{proof}
Formula \eqref{eq:Lap-I} follows directly from \eqref{eq:kappaLap}. For the
Riemann--Liouville-type operator,
\[
\Lap\{\RLTD_{a,b,\tau}u\}(s)
=
s\,\Lap\{k_{a,b,\tau}\}(s)\,\widehat u(s)
=
s\,\frac{\varphi_{a,b,\tau}(s)}{s}\,\widehat u(s)
=
\varphi_{a,b,\tau}(s)\,\widehat u(s),
\]
since \((k_{a,b,\tau}*u)(0)=0\). Finally,
\[
\Lap\{\CTD_{a,b,\tau}u\}(s)
=
\Lap\{k_{a,b,\tau}\}(s)\,\Lap\{u'\}(s)
=
\frac{\varphi_{a,b,\tau}(s)}{s}\,
\bigl(s\widehat u(s)-u(0)\bigr),
\]
which is \eqref{eq:Lap-Caputo}.
\end{proof}

We can finally connect the above definitions through the identity
\begin{equation}
\CTD_{a,b,\tau}u=\RLTD_{a,b,\tau}(u-u(0)) ,
\end{equation}
which shows that the Caputo-type operator is the natural regularization of the Riemann--Liouville-type one, exactly as in the classical fractional setting~\cite{Luchko2021Sonine, Mainardi2022book}.

\subsection{Generalized fundamental theorems}
We now show that the Tricomi integral and derivatives introduced in Sect.~\ref{ss:Tricomi-integral-derivatives} satisfy the generalized fundamental theorems of fractional calculus, in the Sonine-kernel sense.

\begin{theorem}[First fundamental theorem]\label{thm:FFT}
Let \(f\in L^1_{\mathrm{loc}}([0,\infty))\) be causal and of exponential order.
Then, the Riemann--Liouville-type Tricomi derivative is the left inverse of the Tricomi integral operator
\begin{equation}
\RLTD_{a,b,\tau}\,\ITU_{a,b,\tau}f=f.
\label{eq:FFT}
\end{equation}
\end{theorem}

\begin{proof}
Taking Laplace transforms and using \eqref{eq:Lap-I} and \eqref{eq:Lap-RL},
\[
\Lap\{\RLTD_{a,b,\tau}\ITU_{a,b,\tau}f\}(s)
=
\varphi_{a,b,\tau}(s)\,
\TricomiU(a,b,s\tau)\,\widehat f(s)
=
\widehat f(s).
\]
Uniqueness of the Laplace transform gives \eqref{eq:FFT}.
\end{proof}

\begin{theorem}[Second fundamental theorem]\label{thm:SFT}
Let \(u\in AC_{\mathrm{loc}}([0,\infty))\) be causal and of exponential order.
Then, for the Tricomi--Caputo derivative it holds that
\begin{equation}
\ITU_{a,b,\tau}\,\CTD_{a,b,\tau}u=u-u(0).
\label{eq:SFT}
\end{equation}
\end{theorem}

\begin{proof}
Using \eqref{eq:Lap-I} and \eqref{eq:Lap-Caputo},
\[
\Lap\{\ITU_{a,b,\tau}\CTD_{a,b,\tau}u\}(s)
=
\TricomiU(a,b,s\tau)
\left(
\varphi_{a,b,\tau}(s)\widehat u(s)
-
\frac{\varphi_{a,b,\tau}(s)}{s}u(0)
\right)
=
\widehat u(s)-\frac{u(0)}{s}.
\]
Laplace inversion yields \eqref{eq:SFT}.
\end{proof}

\begin{corollary}[Generalized fractional structure]
\label{cor:legit}
The pair \((\ITU_{a,b,\tau},\CTD_{a,b,\tau})\) defines a genuine generalized
fractional calculus in the Sonine-kernel framework. In particular,
\(\CTD_{a,b,\tau}\) is a generalized fractional derivative associated with the
complete Bernstein symbol \(\varphi_{a,b,\tau}\).
\end{corollary}

\begin{proof}
This follows from Theorems \ref{thm:sonine}, \ref{thm:FFT} and \ref{thm:SFT} together with the Laplace-domain formula \eqref{eq:Lap-Caputo}, thanks to the general Sonine-kernel theory \cite{LuchkoYamamoto2020,Luchko2021Sonine,Luchko2021Operational,Luchko2020Fundamental}.
\end{proof}

We emphasize that the reciprocal Tricomi law \eqref{eq:rawsymbol} induces a genuine generalized fractional derivative in the Sonine-kernel sense.
No additional operatorial assumptions are required beyond the Stieltjes and complete Bernstein structure of the Tricomi symbol.

\section{Characterization of the Tricomi family within the Sonine framework}
\label{sec:characterization}

The previous sections show that the Tricomi law generates a genuine Sonine fractional calculus and it would then be natural to ask in what sense this family is intrinsically singled out within the Sonine framework. The answer is not that every Sonine pair with two independent power-law asymptotics must be of Tricomi type, since such a statement would be too strong in general.
What can be proved is a precise rigidity statement inside the Kummer class, indeed among the solutions of the confluent hypergeometric equation, the Tricomi branch is uniquely selected by the Stieltjes property, and the normalization fixes the remaining multiplicative constant. This yields a characterization of the Tricomi--Sonine family inside the class of Sonine pairs whose Laplace-side law belongs to Kummer equation.

We begin with the unscaled Kummer equation
\begin{equation}
zF''(z)+(b-z)F'(z)-aF(z)=0,
\qquad z\in\C\setminus(-\infty,0],
\label{eq:KummerEq}
\end{equation}
with the conditions $a\in(0,1)$, $b\in(1,2)$.
Since \(b\notin\mathbb Z\), a fundamental system of solutions is given by the two linearly independent special functions, namely the Kummer function $M(a,b,z)={}_1F_1(a;b;z)$, coinciding with the confluent hypergeometric function of the first kind, and the Tricomi function $\TricomiU(a,b,z)$ \cite{Abramowitz1965handbook,tricomi1947funzioni,DLMF,Buchholz1969,Temme1996}.

\begin{lemma}[Selection of the Tricomi branch]
\label{lem:branch-selection}
Let \(F\) be analytic in \(\C\setminus(-\infty,0]\) and suppose that \(F\) solves the unscaled Kummer equation \eqref{eq:KummerEq}. Let us also assume that the restriction of \(F\) to \((0,\infty)\) is a Stieltjes function. Then, there exists a constant \(C\ge0\) such that
\begin{equation}
F(z)=C\,\TricomiU(a,b,z),
\qquad z\in\C\setminus(-\infty,0].
\label{eq:FisCU}
\end{equation}
In the case \(F\not\equiv0\), then \(C>0\).
\end{lemma}

\begin{proof}
Since \(b\notin\mathbb Z\), the general solution of \eqref{eq:KummerEq} has the
form
\begin{equation}
F(z)=c_1 M(a,b,z)+c_2 \TricomiU(a,b,z),
\label{eq:general-solution-Kummer}
\end{equation}
for suitable constants \(c_1,c_2\in\C\).

Because \(F\) is Stieltjes on \((0,\infty)\), it admits a representation
\[
F(s)=\frac{\alpha}{s}+\beta+\int_0^\infty \frac{1}{s+r}\,\mu(\dd r),
\qquad s>0,
\]
with \(\alpha,\beta\ge0\) and
\[
\int_0^\infty \frac{1}{1+r}\,\mu(\dd r)<\infty.
\]
In particular, on the one hand, for every \(s\ge1\),
\[
0\le F(s)\le \alpha+\beta+\int_0^\infty \frac{1}{1+r}\,\mu(\dd r),
\]
so \(F\) is bounded on \([1,\infty)\).

On the other hand, the standard large-argument asymptotics give \cite{DLMF,Buchholz1969,Temme1996}
\[
M(a,b,z)\sim \frac{\Gamma(b)}{\Gamma(a)}\,\ee^z z^{a-b},
\qquad z\to+\infty,
\]
whereas
\[
\TricomiU(a,b,z)\sim z^{-a},
\qquad z\to+\infty.
\]
Therefore, if \(c_1\neq0\), the term
\(c_1M(a,b,z)\) would force exponential growth along the positive real axis,
contradicting the boundedness of the Stieltjes function \(F\). Hence \(c_1=0\),
and
\[
F(z)=c_2 \TricomiU(a,b,z).
\]

Since \(F(s)\ge0\) for \(s>0\) and \(\TricomiU(a,b,s)>0\) for \(s>0\), one must
have \(c_2=C\ge0\). If \(F\not\equiv0\), then necessarily \(C>0\).
\end{proof}

\begin{corollary}[Normalization and uniqueness]
\label{cor:normalization-U}
Let \(F\not\equiv0\) satisfy the assumptions of Lemma \ref{lem:branch-selection}.
Then the following statements are equivalent:

\begin{enumerate}
\item[(i)]
\begin{equation}
\lim_{z\to+\infty} z^{a}F(z)=1;
\label{eq:norm-infty}
\end{equation}

\item[(ii)]
\begin{equation}
\lim_{z\to0^+} z^{\,b-1}F(z)=\frac{\Gamma(b-1)}{\Gamma(a)};
\label{eq:norm-zero}
\end{equation}

\item[(iii)]
\begin{equation}
F(z)=\TricomiU(a,b,z)
\qquad\text{for all } z\in\C\setminus(-\infty,0].
\label{eq:exact-U}
\end{equation}
\end{enumerate}
\end{corollary}

\begin{proof}
By Lemma \ref{lem:branch-selection}, one has \(F(z)=C\TricomiU(a,b,z)\) with
\(C>0\). Since
\[
\TricomiU(a,b,z)\sim z^{-a},
\qquad z\to+\infty,
\]
condition \eqref{eq:norm-infty} is equivalent to \(C=1\). Likewise, since
\[
\TricomiU(a,b,z)\sim
\frac{\Gamma(b-1)}{\Gamma(a)}\,z^{1-b},
\qquad z\to0^+,
\]
condition \eqref{eq:norm-zero} is also equivalent to \(C=1\). This proves the
equivalence with \eqref{eq:exact-U}.
\end{proof}

We now translate the previous rigidity statement to the Sonine setting.

\begin{theorem}[Characterization of the Tricomi--Sonine family]
\label{thm:TricomiSonineCharacterization}
Let \((\kappa,k)\) be a Sonine pair on \((0,\infty)\), that is,
\begin{equation}
(\kappa*k)(t)=1,
\qquad t>0,
\label{eq:Sonine-char}
\end{equation}
and denote
\[
K(s):=\Lap\{\kappa\}(s),
\qquad s>0.
\]
Let us assume that \(K\) is a Stieltjes function and that there exists \(\tau>0\) such
that the rescaled function
\begin{equation}
F(z):=K\!\left(\frac{z}{\tau}\right),
\qquad z>0,
\label{eq:FfromK}
\end{equation}
extends analytically to \(\C\setminus(-\infty,0]\) and solves the Kummer equation \eqref{eq:KummerEq}.

Then, there exists a constant \(C>0\) such that
\begin{equation}
K(s)=C\,\TricomiU(a,b,s\tau),
\qquad s>0,
\label{eq:K-CU}
\end{equation}
and consequently
\begin{equation}
\kappa(t)=C\,\kappa_{a,b,\tau}(t),
\qquad
k(t)=C^{-1}k_{a,b,\tau}(t),
\qquad t>0,
\label{eq:kappa-k-scaled}
\end{equation}
where \(\kappa_{a,b,\tau}\) and \(k_{a,b,\tau}\) are the normalized Tricomi kernels introduced in Sect.~\ref{sec:sonine}.

If, in addition, one of the equivalent normalization conditions
\begin{equation}
\lim_{s\to+\infty}(s\tau)^aK(s)=1
\label{eq:Knorm1}
\end{equation}
or
\begin{equation}
\lim_{s\to0^+}(s\tau)^{b-1}K(s)=\frac{\Gamma(b-1)}{\Gamma(a)}
\label{eq:Knorm2}
\end{equation}
holds, then \(C=1\), and therefore
\begin{equation}
K(s)=\TricomiU(a,b,s\tau),
\qquad
\kappa(t)=\kappa_{a,b,\tau}(t),
\qquad
k(t)=k_{a,b,\tau}(t).
\label{eq:exact-tricomi-sonine}
\end{equation}
\end{theorem}

\begin{proof}
Since \(K\) is Stieltjes on \((0,\infty)\), the same is true for $F(z)$ in \eqref{eq:FfromK}.
By assumption, \(F\) extends analytically to the slit plane and solves
\eqref{eq:KummerEq}. Therefore Lemma \ref{lem:branch-selection} yields
\[
F(z)=C\,\TricomiU(a,b,z)
\]
for some \(C>0\). Replacing \(z\) by \(s\tau\), we obtain
\[
K(s)=F(s\tau)=C\,\TricomiU(a,b,s\tau),
\]
which proves \eqref{eq:K-CU}.

Taking inverse Laplace transforms and using Proposition \ref{prop:kappa}, we get
\[
\kappa(t)=C\,\kappa_{a,b,\tau}(t).
\]
Since \((\kappa,k)\) is a Sonine pair, its Laplace-side identity is
\[
\Lap\{\kappa\}(s)\Lap\{k\}(s)=\frac{1}{s}.
\]
Hence
\[
\Lap\{k\}(s)=\frac{1}{sK(s)}
=
\frac{1}{Cs\,\TricomiU(a,b,s\tau)}
=
\frac{1}{C}\Lap\{k_{a,b,\tau}\}(s),
\]
which yields
\[
k(t)=C^{-1}k_{a,b,\tau}(t),
\]
proving so \eqref{eq:kappa-k-scaled}.

Finally, conditions \eqref{eq:Knorm1} and \eqref{eq:Knorm2} are precisely the rescaled versions of \eqref{eq:norm-infty} and \eqref{eq:norm-zero}. Thus, thanks to Corollary \ref{cor:normalization-U}, each of them is equivalent to \(C=1\), which gives \eqref{eq:exact-tricomi-sonine}.
\end{proof}

As a consequence of Theorem~\ref{thm:TricomiSonineCharacterization}, the induced Sonine fractional calculus coincides, modulo a scalar normalization, with the Tricomi fractional calculus at the operator level. Under either normalization
\eqref{eq:Knorm1} or \eqref{eq:Knorm2}, one has
\begin{equation}
I_\kappa=\ITU_{a,b,\tau},
\qquad
D_k^{C}=\CTD_{a,b,\tau}.
\end{equation}
This follows immediately from \eqref{eq:kappa-k-scaled} and from the
definitions of the associated convolution operators.

The previous result is an internal characterization of the Tricomi family within the Sonine framework, restricted to the Kummer class. It does not claim that every Sonine pair with two distinct algebraic asymptotic exponents must be of Tricomi type.

Conceptually, the characterization has two layers. First, the Stieltjes property excludes the Kummer \(M\)-branch and selects the Tricomi branch.
Second, the asymptotic normalization fixes the residual multiplicative freedom. In this sense, the normalized Tricomi law is not merely an explicit example in the Sonine framework, but the unique normalized Stieltjes member of the Kummer class.

\section{Asymptotic identifiability within the Tricomi-\(U\) family}\label{sec:ident}

The asymptotic formulas obtained above imply that the parameters of the raw Tricomi symbol $\varphi_{a,b,\tau}$ can be recovered from its two asymptotic sectors and equivalently from the associated L\'evy density or hereditary kernel.

To summarize the asymptotic behavior, let collect \eqref{eq:rawlarge}, \eqref{eq:rawsmall},
\eqref{eq:msmall}, \eqref{eq:mlarge},
\eqref{eq:ksmall}, and \eqref{eq:klarge} into the following proposition.
\begin{proposition}\label{prop:collect-asym}
For every \(a\in(0,1)\), \(b\in(1,2)\) and \(\tau>0\), it holds
\[
\varphi_{a,b,\tau}(s)\overset{s\to\infty}{\sim} \tau^a s^a \,,
\qquad
\varphi_{a,b,\tau}(s)\overset{s\to 0^+}{\sim}
\frac{\Gamma(a)}{\Gamma(b-1)}\,\tau^{b-1}s^{b-1} \,,
\]
\[
m_{a,b,\tau}(r)\overset{r\to 0^+}{\sim}
\frac{a\,\tau^a}{\Gamma(1-a)}\,r^{-1-a}\,,
\qquad
m_{a,b,\tau}(r)\overset{r\to\infty}{\sim}
\frac{\Gamma(a)}{\Gamma(b-1)}
\frac{(b-1)\tau^{b-1}}{\Gamma(2-b)}
\,r^{-b} \,,
\]
and
\[
k_{a,b,\tau}(t)\overset{t \to 0^+}{\sim}
\frac{\tau^a}{\Gamma(1-a)}\,t^{-a} \,,
\qquad
k_{a,b,\tau}(t)\overset{t\to\infty}{\sim}
\frac{\Gamma(a)}{\Gamma(b-1)\Gamma(2-b)}
\,\tau^{b-1} t^{1-b} \,.
\]
\end{proposition}

\begin{theorem}[Asymptotic identifiability]\label{thm:ident-family}
Let $\varphi_{a,b,\tau}$ and $\varphi_{\tilde a,\tilde b,\tilde\tau}$ be two raw Tricomi symbols with $a,\tilde a\in(0,1)$, $b,\tilde b\in(1,2)$ and $\tau,\tilde\tau>0$.

Let us assume that
\begin{equation}
\frac{\varphi_{a,b,\tau}(s)}{\varphi_{\tilde a,\tilde b,\tilde\tau}(s)}
\longrightarrow C_\infty\in(0,\infty),
\qquad s\to\infty,
\label{eq:ratio-infty}
\end{equation}
and
\begin{equation}
\frac{\varphi_{a,b,\tau}(s)}{\varphi_{\tilde a,\tilde b,\tilde\tau}(s)}
\longrightarrow C_0\in(0,\infty),
\qquad s\to0^+.
\label{eq:ratio-zero}
\end{equation}
Then, necessarily
\[
a=\tilde a,
\qquad
b=\tilde b,
\]
and moreover,
\[
C_\infty=\frac{\tau^a}{\tilde\tau^{\,a}},
\qquad
C_0=
\frac{\Gamma(a)/\Gamma(b-1)}{\Gamma(\tilde a)/\Gamma(\tilde b-1)}
\frac{\tau^{b-1}}{\tilde\tau^{\,b-1}}.
\]
In particular, if \(C_\infty=C_0=1\), then \(\tau=\tilde\tau\) and consequently
\[
\varphi_{a,b,\tau}\equiv \varphi_{\tilde a,\tilde b,\tilde\tau}.
\]
\end{theorem}

\begin{proof}
By the asymptotics in Proposition~\ref{prop:collect-asym},
\[
\varphi_{a,b,\tau}(s)\sim \tau^a s^a,
\qquad
\varphi_{\tilde a,\tilde b,\tilde\tau}(s)\sim
\tilde\tau^{\,\tilde a}s^{\tilde a},
\qquad s\to\infty,
\]
so that
\[
\frac{\varphi_{a,b,\tau}(s)}{\varphi_{\tilde a,\tilde b,\tilde\tau}(s)}
\sim
\frac{\tau^a}{\tilde\tau^{\,\tilde a}}\,s^{a-\tilde a}.
\]
A finite nonzero limit at infinity forces \(a=\tilde a\), and then yields the
formula for \(C_\infty\).

Likewise, for $s\to 0^+$ we have
\[
\varphi_{a,b,\tau}(s)\sim
\frac{\Gamma(a)}{\Gamma(b-1)}\,\tau^{b-1}s^{b-1},
\qquad
\varphi_{\tilde a,\tilde b,\tilde\tau}(s)\sim
\frac{\Gamma(\tilde a)}{\Gamma(\tilde b-1)}\,
\tilde\tau^{\,\tilde b-1}s^{\tilde b-1},
\]
leading to
\[
\frac{\varphi_{a,b,\tau}(s)}{\varphi_{\tilde a,\tilde b,\tilde\tau}(s)}
\sim
\frac{\Gamma(a)/\Gamma(b-1)}{\Gamma(\tilde a)/\Gamma(\tilde b-1)}
\frac{\tau^{b-1}}{\tilde\tau^{\,\tilde b-1}}
\,s^{b-\tilde b}, \qquad s\to0^+.
\]
As before, a finite nonzero limit at the origin forces \(b=\tilde b\) and then gives the formula for \(C_0\). If \(C_\infty=C_0=1\), then \(a=\tilde a\) and \(b=\tilde b\) and from \(C_\infty=1\) we obtain \(\tau=\tilde\tau\).
\end{proof}

\begin{theorem}[Recovery formulas]\label{thm:recovery}
Let \(a\in(0,1)\), \(b\in(1,2)\) and \(\tau>0\). Then, we have
\begin{equation}
a
=
\lim_{s\to\infty}\frac{\log \varphi_{a,b,\tau}(s)}{\log s}
=
-1-\lim_{r\to0^+}\frac{\log m_{a,b,\tau}(r)}{\log r}
=
-\lim_{t\to0^+}\frac{\log k_{a,b,\tau}(t)}{\log t},
\label{eq:a-ident}
\end{equation}
and
\begin{equation}
b
=
1+\lim_{s\to0^+}\frac{\log \varphi_{a,b,\tau}(s)}{\log s}
=
-\lim_{r\to\infty}\frac{\log m_{a,b,\tau}(r)}{\log r}
=
1-\lim_{t\to\infty}\frac{\log k_{a,b,\tau}(t)}{\log t}.
\label{eq:b-ident}
\end{equation}
Equivalently,
\begin{equation}
a
=
-\lim_{s\to\infty}\frac{\log \TricomiU(a,b,s\tau)}{\log s},
\qquad
b
=
1-\lim_{s\to0^+}\frac{\log \TricomiU(a,b,s\tau)}{\log s}.
\label{eq:U-ident}
\end{equation}
\end{theorem}

\begin{proof}
All identities follow directly from the asymptotic formulas collected in Proposition \ref{prop:collect-asym} together with \eqref{eq:Usmall} and \eqref{eq:Ularge}.
For instance, \eqref{eq:rawlarge} gives
\[
\log \varphi_{a,b,\tau}(s)=a\log s+O(1),
\qquad s\to\infty,
\]
hence
\[
\lim_{s\to\infty}\frac{\log \varphi_{a,b,\tau}(s)}{\log s}=a.
\]
The remaining formulas are obtained in exactly the same way.
\end{proof}

As a consequence, rearranging the leading coefficients in
\eqref{eq:rawlarge}, \eqref{eq:rawsmall}, \eqref{eq:msmall} and
\eqref{eq:klarge}, we might also state the following corollary to set $\tau$.
\begin{corollary}[Recovery of the scale parameter]\label{cor:tau-recovery}
After identifying \(a\) and \(b\), the scale parameter \(\tau\) is recovered from the asymptotic amplitudes. More precisely,
\begin{equation}
\tau
=
\left(
\lim_{s\to\infty}\frac{\varphi_{a,b,\tau}(s)}{s^a}
\right)^{1/a}
=
\left(
\frac{\Gamma(1-a)}{a}
\lim_{r\to0^+} r^{1+a}m_{a,b,\tau}(r)
\right)^{1/a},
\label{eq:tau-fast}
\end{equation}
and also
\begin{align}
\tau
&=
\left(
\frac{\Gamma(b-1)}{\Gamma(a)}
\lim_{s\to0^+}\frac{\varphi_{a,b,\tau}(s)}{s^{b-1}}
\right)^{1/(b-1)}
\\
&=
\left(
\frac{\Gamma(b-1)\Gamma(2-b)}{\Gamma(a)}
\lim_{t\to\infty} t^{b-1}k_{a,b,\tau}(t)
\right)^{1/(b-1)}.
\label{eq:tau-slow}
\end{align}
\end{corollary}

\section{The Volterra formulation and the associated Cauchy problem}
\label{sec:model}

The natural low-regularity evolution problem induced by the Tricomi Sonine pair is the Volterra equation of the second kind~\cite{GripenbergLondenStaffans1990}
\begin{equation}
u(t)+\lambda\int_0^t \kappa_{a,b,\tau}(t-\xi)\,u(\xi)\,\dd\xi
=
u_0+\int_0^t \kappa_{a,b,\tau}(t-\xi)\,f(\xi)\,\dd\xi,
\quad t>0,
\label{eq:Volterra}
\end{equation}
where \(\lambda>0\), \(u_0\in\R\) and \(f\) is a causal function. This formulation is well defined under minimal local integrability assumptions. The differential Cauchy problem associated with the Tricomi--Caputo derivative appears as a more regular realization of the same dynamics. For the general viewpoint on fractional evolution equations driven by complete Bernstein symbols
\cite{Kochubei2011,LuchkoYamamoto2020,Luchko2021Operational}.

\begin{definition}[Mild solution]
Let us consider \(u_0\in\R\), \(\lambda>0\) and \(f\in L^1_{\mathrm{loc}}([0,\infty))\).
A causal function \(u\in L^1_{\mathrm{loc}}([0,\infty))\) is called a mild solution of \eqref{eq:Volterra} if it satisfies \eqref{eq:Volterra} for almost every \(t>0\).
\end{definition}

\begin{theorem}[Existence and uniqueness of mild solutions]
\label{thm:ExistUniqueMild}
Let \(u_0\in\R\), \(\lambda>0\) and \(f\in L^1_{\mathrm{loc}}([0,\infty))\).
Then, the Volterra equation \eqref{eq:Volterra} admits a unique mild solution
\(u\in L^1_{\mathrm{loc}}([0,\infty))\).

More precisely, for every \(T>0\), the restriction \(u|_{(0,T)}\) is the unique
solution in \(L^1(0,T)\) of
\begin{equation}
u+\lambda V_Tu=g_T,
\label{eq:VT-equation}
\end{equation}
where we defined
\begin{equation}
(V_Tu)(t):=\int_0^t \kappa_{a,b,\tau}(t-\xi)\,u(\xi)\,\dd\xi,    
\end{equation}
and
\begin{equation}
g_T(t):=u_0+\int_0^t \kappa_{a,b,\tau}(t-\xi)\,f(\xi)\,\dd\xi.    
\end{equation}

\end{theorem}

\begin{proof}
Let us fix \(T>0\). Since \(\kappa_{a,b,\tau}\in L^1(0,T)\) by Proposition \ref{prop:kappa} and given \(f\in L^1(0,T)\), one has \(g_T\in L^1(0,T)\). The operator \(V_T:L^1(0,T)\to L^1(0,T)\) is bounded, and if
\[
K_T:=\int_0^T \kappa_{a,b,\tau}(t)\,\dd t,
\]
then, since the iterates of a Volterra operator involve integration over ordered time variables, one obtains the standard factorial decay for its powers \cite{GripenbergLondenStaffans1990},
\[
\|V_T^n\|_{\mathcal{L}(L^1(0,T))}\le \frac{K_T^n}{n!},
\qquad n\ge1.
\]
Therefore, the Neumann series
\[
(I+\lambda V_T)^{-1}
=
\sum_{n=0}^\infty (-\lambda)^n V_T^n
\]
converges in operator norm on \(L^1(0,T)\). Consequently, in \(L^1(0,T)\), \eqref{eq:VT-equation} admits a unique solution
\[
u_T=(I+\lambda V_T)^{-1}g_T .
\]

If \(0<T_1<T_2\), uniqueness on \((0,T_1)\) implies that the solutions
\(u_{T_1}\) and \(u_{T_2}\) coincide on \((0,T_1)\). Thus the family
\(\{u_T\}_{T>0}\) is consistent and defines a unique function \(u\in L^1_{\mathrm{loc}}([0,\infty))\) satisfying \eqref{eq:Volterra} on every finite interval and this is the unique mild solution.
\end{proof}

This latter Theorem \ref{thm:ExistUniqueMild} provides the primary well-posedness result for the Tricomi-induced evolution problem and we emphasize that, at this level, no absolute continuity of the solution is required.

\begin{proposition}[Resolvent kernel for the forcing term]
\label{prop:forcingresolvent}
Let us define
\begin{equation}
\widehat P_{\lambda,a,b,\tau}(s)
:=
\frac{1}{\lambda+\varphi_{a,b,\tau}(s)}
=
\frac{\TricomiU(a,b,s\tau)}{1+\lambda\,\TricomiU(a,b,s\tau)},
\qquad s>0.
\label{eq:Phat}
\end{equation}
Then, \(\widehat P_{\lambda,a,b,\tau}\) is a Stieltjes function and thus \(P_{\lambda,a,b,\tau}\) is completely monotone. Moreover, \(P_{\lambda,a,b,\tau}\) satisfies the resolvent identity
\begin{equation}
P_{\lambda,a,b,\tau}
+
\lambda\,\kappa_{a,b,\tau}*P_{\lambda,a,b,\tau}
=
\kappa_{a,b,\tau}.
\label{eq:Presolvent}
\end{equation}
\end{proposition}

\begin{proof}
Recalling that \(\varphi_{a,b,\tau}\) is a complete Bernstein function by Corollary \ref{cor:rawCBF}, the shifted function
\[
s\mapsto \lambda+\varphi_{a,b,\tau}(s)
\]
is again complete Bernstein \cite{SchillingSongVondracek2012}. By Proposition \ref{prop:StCBF}, its reciprocal \(\widehat P_{\lambda,a,b,\tau}\) is a Stieltjes function, hence the Laplace transform of a completely monotone kernel.

Furthermore, we find
\[
\bigl(1+\lambda\TricomiU(a,b,s\tau)\bigr)\widehat P_{\lambda,a,b,\tau}(s)
=
\TricomiU(a,b,s\tau),
\]
and since \(\Lap\{\kappa_{a,b,\tau}\}(s)=\TricomiU(a,b,s\tau)\), Laplace inversion yields \eqref{eq:Presolvent}.
\end{proof}

\begin{theorem}[Laplace-resolvent characterization of mild solutions]
\label{thm:ResolventRepresentation}
Assume that \(f\) is causal and of exponential order and let also \(u\) be a causal function of exponential order. Then \(u\) is a mild solution of
\eqref{eq:Volterra} if and only if
\begin{equation}
\widehat u(s)
=
\widehat S_{\lambda,a,b,\tau}(s)\,u_0
+
\widehat P_{\lambda,a,b,\tau}(s)\,\widehat f(s),
\qquad s>0,
\label{eq:uhat-resolvent}
\end{equation}
where
\begin{equation}
\widehat S_{\lambda,a,b,\tau}(s)
:=
\frac{1}{s\bigl(1+\lambda\,\TricomiU(a,b,s\tau)\bigr)}
=
\frac{\varphi_{a,b,\tau}(s)}
{s\bigl(\lambda+\varphi_{a,b,\tau}(s)\bigr)},
\label{eq:Shat}
\end{equation}
and \(\widehat P_{\lambda,a,b,\tau}\) is given by \eqref{eq:Phat}.
Equivalently,
\begin{equation}
u(t)=S_{\lambda,a,b,\tau}(t)\,u_0
+
\int_0^t P_{\lambda,a,b,\tau}(t-\xi)\,f(\xi)\,\dd\xi.
\label{eq:u-resolvent}
\end{equation}
\end{theorem}

\begin{proof}
Taking Laplace transforms in \eqref{eq:Volterra} and using
\(\Lap\{\kappa_{a,b,\tau}\}(s)=\TricomiU(a,b,s\tau)\), we obtain
\[
\widehat u(s)+\lambda\,\TricomiU(a,b,s\tau)\widehat u(s)
=
\frac{u_0}{s}+\TricomiU(a,b,s\tau)\widehat f(s).
\]
Hence, it follows that
\[
\widehat u(s)
=
\frac{1}{s\bigl(1+\lambda\,\TricomiU(a,b,s\tau)\bigr)}\,u_0
+
\frac{\TricomiU(a,b,s\tau)}{1+\lambda\,\TricomiU(a,b,s\tau)}\,\widehat f(s),
\]
which is exactly \eqref{eq:uhat-resolvent} and the equivalence with
\eqref{eq:u-resolvent} follows by Laplace inversion.
\end{proof}

Interestingly, from \eqref{eq:Shat} and \eqref{eq:Phat}, employing \eqref{eq:intro:rawsymbol}, it is quite straightforward to write
\begin{equation}
\widehat S_{\lambda,a,b,\tau}(s) =
\frac{1}{s}-\frac{\lambda}{s}\widehat P_{\lambda,a,b,\tau}(s),    
\end{equation}
and taking Laplace inverse transforms leads to the following statement.
\begin{corollary}[Relation between the two resolvent kernels]
\label{cor:SfromP}
Let us consider the two resolvent kernels
\(P_{\lambda,a,b,\tau}\) and \(S_{\lambda,a,b,\tau}\).
Then, \(S_{\lambda,a,b,\tau}\) is related to \(P_{\lambda,a,b,\tau}\) by
\begin{equation}
S_{\lambda,a,b,\tau}(t)
=
1-\lambda\int_0^t P_{\lambda,a,b,\tau}(\xi)\,\dd\xi,
\qquad t\ge0.
\label{eq:SfromP}
\end{equation}
\end{corollary}

In other words, \(P_{\lambda,a,b,\tau}\) is the resolvent kernel naturally associated with the forcing term, whereas \(S_{\lambda,a,b,\tau}\) propagates the initial datum. In particular, the homogeneous solution
\begin{equation}
u_h(t)=S_{\lambda,a,b,\tau}(t)u_0    
\end{equation}
is the natural scalar relaxation law associated with the shifted symbol \(\lambda+\varphi_{a,b,\tau}\).

This notion of mild solution is consistent with the general well-posedness theory for Caputo-type fractional differential equations developed in \cite{Diethelm2019}.

We now move to the differential realization of the same dynamics, considering the associated scalar Cauchy problem
\begin{equation}
\CTD_{a,b,\tau}u(t)+\lambda u(t)=f(t),
\qquad t>0,
\label{eq:Cauchy}
\end{equation}
with initial condition $u(0)=u_0$.

\begin{definition}[Classical solution]
Let \(u_0\in\R\), \(\lambda>0\) and \(f\) be causal. Then, a causal function \(u\in AC_{\mathrm{loc}}([0,\infty))\) of exponential order is called a classical solution of the Cauchy problem \eqref{eq:Cauchy}, if \(\CTD_{a,b,\tau}u\) is well defined almost everywhere on \((0,\infty)\), the identity \eqref{eq:Cauchy} holds
almost everywhere and \(u(0)=u_0\).
\end{definition}

\begin{theorem}[Equivalence with the classical Cauchy problem]
\label{thm:VolterraEquiv}
Let us assume that $f$ is causal and of exponential order.

\begin{enumerate}
\item[(i)]
If \(u\in AC_{\mathrm{loc}}([0,\infty))\) is a causal classical solution of \eqref{eq:Cauchy} and it is of exponential order, then \(u\) is a mild solution of \eqref{eq:Volterra}.
\item[(ii)]
If \(u\) is a mild solution of \eqref{eq:Volterra}, belongs to \(AC_{\mathrm{loc}}([0,\infty))\) and it is of exponential order, then \(u\) is a classical solution of \eqref{eq:Cauchy}.
\end{enumerate}
\end{theorem}

\begin{proof}
Assume first that \(u\) is a classical solution of \eqref{eq:Cauchy}. Applying the Tricomi integral \(\ITU_{a,b,\tau}\) to both sides, using Theorem \ref{thm:SFT} and \eqref{eq:ITU} gives
\[
u(t)-u(0)+\lambda\int_0^t \kappa_{a,b,\tau}(t-\xi)\,u(\xi)\,\dd\xi
=
\int_0^t \kappa_{a,b,\tau}(t-\xi)\,f(\xi)\,\dd\xi.
\]
Since \(u(0)=u_0\), this is exactly \eqref{eq:Volterra}, so \(u\) is a mild
solution.

Conversely, assume that \(u\) is a mild solution of \eqref{eq:Volterra} and
that \(u\in AC_{\mathrm{loc}}([0,\infty))\) is of exponential order. Rewriting
\eqref{eq:Volterra}, we obtain
\[
u(t)-u_0
=
\int_0^t \kappa_{a,b,\tau}(t-\xi)\,[f(\xi)-\lambda u(\xi)]\,\dd\xi
=
(\ITU_{a,b,\tau}(f-\lambda u))(t).
\]
Applying \(\RLTD_{a,b,\tau}\) and using Theorem \ref{thm:FFT} yields
\[
\RLTD_{a,b,\tau}(u-u_0)=f-\lambda u.
\]
Since \(u\in AC_{\mathrm{loc}}([0,\infty))\), the following identity holds
\[
\CTD_{a,b,\tau}u=\RLTD_{a,b,\tau}(u-u(0)) ,
\]
while \(u(0)=u_0\) follows from \eqref{eq:Volterra} by letting
\(t\to0^+\), because both convolution terms vanish at the origin. Therefore, we find
\[
\CTD_{a,b,\tau}u+\lambda u=f,
\]
proving that \(u\) is a classical solution of \eqref{eq:Cauchy}.
\end{proof}

Thus, the Volterra formulation plays a primary role under minimal assumptions on the data and on the solution. In this setting, the evolution problem is naturally interpreted in terms of mild solutions, which are well defined in the purely hereditary framework and do not require any differentiability in time.

The differential scalar Cauchy problem is recovered from the Volterra formulation only under additional regularity assumptions. In particular, if the mild solution is absolutely continuous, so that the Tricomi--Caputo derivative is well defined, then the two formulations are equivalent and describe the same dynamics.

\subsection{Shifted symbols and bounded relaxation laws}
\label{subsec:shifted}

As already seen in the demonstration of Proposition \ref{prop:forcingresolvent}, the Volterra formulation naturally introduces shifted symbols of the form
$\lambda+\varphi_{a,b,\tau}(s)$ for $\lambda>0$, through the scalar evolution problem \eqref{eq:Cauchy}. The corresponding
homogeneous dynamics is governed by the propagator
\(S_{\lambda,a,b,\tau}\), whose Laplace transform is $\widehat S_{\lambda,a,b,\tau}(s)$ in~\eqref{eq:Shat}.

Closely related to this propagator, we have the bounded relaxation law
\begin{equation}
F_{\lambda,a,b}(s\tau)
:=
\frac{\TricomiU(a,b,s\tau)}{1+\lambda\,\TricomiU(a,b,s\tau)}
=
\frac{1}{\lambda+\varphi_{a,b,\tau}(s)},
\qquad s>0.
\label{eq:Flambda}
\end{equation}
Since \(\varphi_{a,b,\tau}\) is a complete Bernstein function, the shifted symbol \(\lambda+\varphi_{a,b,\tau}\) is again complete Bernstein, so that \(F_{\lambda,a,b}\) is a Stieltjes function \cite{SchillingSongVondracek2012} and, in particular,
\(F_{\lambda,a,b}\) is completely monotone in the time domain.

For \(\lambda=1\), one recovers the normalized bounded Tricomi law, i.~e.
\begin{equation}
F_{1,a,b}(s\tau) := F_{a,b}(s\tau)
=
\frac{\TricomiU(a,b,s\tau)}{1+\TricomiU(a,b,s\tau)},
\label{eq:Fbounded}
\end{equation}
showing that the bounded law can be naturally interpreted within the Volterra framework as the bounded response governing the homogeneous dynamics associated with the shifted generator \(1+\varphi_{a,b,\tau}\), rather than as an external modification of the raw hereditary symbol.

Moreover, using the asymptotics of \(\TricomiU(a,b,s\tau)\) \eqref{eq:Usmall} and \eqref{eq:Ularge}, one obtains respectively
\begin{equation}
F_{a,b}(s\tau)\sim
1-\frac{\Gamma(a)}{\Gamma(b-1)}(s\tau)^{b-1},
\qquad s\to0^+,
\label{eq:Fsmall}
\end{equation}
and
\begin{equation}
F_{a,b}(s\tau)\sim
(s\tau)^{-a},
\qquad |s|\to\infty,\qquad \Re\{s\}>0.
\label{eq:Flarge}
\end{equation}
As a consequence, the bounded law preserves the high-frequency algebraic signature of the raw Tricomi symbol, while saturating at low frequency.

\section{Discussion and concluding remarks}\label{sec:discussion}

The results of this paper identify the Tricomi branch as a distinguished law within the Sonine framework. In the admissible parameters range \(a\in(0,1)\), \(b\in(1,2)\) and \(\tau>0\), the function \(s\mapsto \TricomiU(a,b,s\tau)\) is Stieltjes and its reciprocal is therefore a complete Bernstein function. This yields, in a canonical way, the associated Sonine pair together with the corresponding Tricomi integral and the related Riemann--Liouville-type and Caputo-type derivatives.

A notable feature of the resulting calculus is its genuinely two-scale structure. The exponents \(a\) and \(b-1\) govern different asymptotic regimes and appear consistently in the Laplace symbol, in the L\'evy density and in the hereditary kernel, showing that the two-scale behavior is intrinsic to the construction itself. Within the Kummer class, this structure is rigid enough to select the Tricomi branch, once compatibility with the Stieltjes/Sonine setting and the natural asymptotic normalization are imposed.

The same analytic framework also leads naturally to the associated evolution problem. The Tricomi symbol generates a well-posed Volterra equation under minimal local assumptions, while the corresponding scalar Cauchy problem driven by the Tricomi--Caputo derivative is recovered under additional regularity. In this context, the bounded Tricomi law emerges as the relaxation law associated with shifted complete Bernstein symbols, arising so from the dynamics itself rather than from an external modification of the hereditary kernel.

Another relevant point is that no additional regularization is required at the level of the Tricomi--Caputo derivative. The singular behaviour of the derivative kernel remains fully compatible with the Sonine framework, and the boundedness of the corresponding relaxation law follows from the shifted-symbol construction without altering the original memory structure. Also, the Tricomi–$U$ construction highlights the flexibility of the Sonine framework in accommodating fractional-type dynamics generated by non-power-law kernels.

These results should also be viewed in connection with recent parallel developments on Tricomi-based fractional models, where complementary aspects such as resolvent families and operator-theoretic properties are investigated \cite{tudelapi2026passive2plateau}. In this sense, the present paper provides the explicit Stieltjes and Sonine foundation underlying that broader line of research.

{\color{black}{

From a physical and practical perspective, the choice of the Tricomi function as a fractional kernel is strongly motivated by multi-scale scenarios in continuum mechanics and anomalous transport. In complex media, such as generalized viscoelastic fluids or heterogeneous porous networks, a single fractional exponent often fails to capture the transition from short-time (fast) elastic responses to long-time (slow) viscous relaxation. The Tricomi kernel inherently overcomes this limitation by naturally embedding two decoupled physical scales. This allows the proposed framework to describe complex attenuation laws and multi-regime diffusion processes directly within a single operator, avoiding the need for multi-term or distributed-order differential formulations.
The Tricomi law also identifies a bounded passive relaxation model for physical impedance and relaxation phenomena, including broadband dielectric and electrochemical impedance data, where Debye-type \cite{debye1929} and Cole-Cole-type \cite{cole1941} responses can be recovered, as show explicitly in \cite{tudelapi2026passive2plateau}. That bounded model corresponds precisely to the $\lambda=1$ shifted generator analysed in Sect. \ref{subsec:shifted}, so the construction developed here directly contains and explains the law used in those applications. The bounded Tricomi law then appears naturally as the relaxation law associated with shifted generators.

On a more theoretical note, the analytical structure of the Volterra operator studied in Sect. \ref{sec:model} opens up interesting questions regarding its spectral properties. While a comprehensive spectral characterization falls outside the scope of this initial study, it is worth noting that due to the Stieltjes and complete Bernstein nature of our symbols, the spectrum of the associated integral operator is intimately connected to the branch cut along the negative real axis. The explicit algebraic representation of the Tricomi function implies that the operator does not possess a discrete point spectrum in standard $L^2$ spaces, but rather a continuous spectrum governed by the boundary values derived in Lemma \ref{lem:cut}. A rigorous functional-analytic investigation of this spectrum and its associated semigroup generation properties represents a promising direction for future research.
}}

%\section{Conclusions}\label{sec:conclusions}

%We have shown that the Tricomi confluent hypergeometric function induces a canonical generalized fractional calculus in the Sonine kernel framework. In the suitable range of the parameters $a$, $b$ and $\tau$, the function $\TricomiU(a,b,s\tau)$ is a Stieltjes law, its reciprocal is a complete Bernstein function and this analytic structure uniquely determines the associated Sonine pair together with the corresponding fractional integral and derivative operators.

%Within this setting,  we derived the associated Volterra formulation, proved well-posedness of mild solutions, and recovered the corresponding scalar Cauchy problem under additional regularity assumptions.

In conclusion, the Tricomi function provides an explicit special-function model in which Stieltjes structure, complete Bernstein geometry, Sonine calculus and two-scale asymptotics coexist in a single coherent framework. This makes the Tricomi family a natural candidate for further developments in generalized fractional dynamics and related nonlocal evolution problems.

\backmatter

\bmhead{Acknowledgements}
The work of I.~C. has been carried out in the framework of the activities of the Italian National Group of Mathematical Physics (GNFM), INdAM.

%\section*{Declarations}

%\noindent\textbf{Funding}  
%This work received no specific grant from any funding agency in the public, commercial, or not-for-profit sectors.

\medskip
\noindent\textbf{Conflict of interest}  
The authors declare that they have no conflict of interest.

%\medskip
%\noindent\textbf{Availability of data and materials}  
No datasets were generated or analysed during the current study.

%\medskip
%\noindent\textbf{Code availability}  
%Not applicable.

%\medskip
%\noindent\textbf{Authors' contributions}  
%Both authors contributed to the conception and development of the work. Ivano Colombaro led the mathematical analysis. Marc Tudela-Pi contributed to the mathematical development, interpretation, and overall structuring of the manuscript. Both authors wrote, revised and approved the final manuscript.

%\bibliographystyle{spmpsci}
\bibliography{fcaa-bibliography}
\end{document}